\documentclass{article}

\usepackage{a4wide}

\usepackage{amssymb}
\usepackage{amsmath}

\usepackage{microtype}

\usepackage{url}
\usepackage{proof}

\usepackage[colorinlistoftodos]{todonotes} 

\newcommand{\FFsymb}{{\sf FF}}

\newcommand{\FFpsymb}[1]{\FFsymb\kern-0.152em_{#1}}
\newcommand{\FFp}[2]{\FFpsymb{#1}(#2)}
\newcommand{\FFcansymb}{\FFpsymb\binfcanindex}
\newcommand{\FFcan}[1]{\FFp\binfcanindex #1}

\newcommand{\NFFsymb}{\mbox{$\not\kern-0.152em{\sf FF}$}}

\newcommand{\NFFpsymb}[1]{\NFFsymb\kern-0.152em_{#1}}
\newcommand{\NFFp}[2]{\NFFpsymb{#1}(#2)}
\newcommand{\NFFcansymb}{\NFFpsymb\binfcanindex}
\newcommand{\NFFcan}[1]{\NFFp\binfcanindex #1}

\newcommand{\NFPsymb}{\mbox{$\not\kern-0.152em{\sf FPP}$}}

\newcommand{\finfinsymb}{\mathsf{finfin}}
\newcommand{\finfin}[1]{\finfinsymb(#1)}
\newcommand{\inffinsymb}{\mathsf{inffin}}
\newcommand{\inffin}[1]{\inffinsymb(#1)}
\newcommand{\finfinmemsymb}{\mathsf{finfinext}}
\newcommand{\finfinmem}[1]{\finfinmemsymb(#1)}
\newcommand{\inffinmemsymb}{\mathsf{inffinext}}
\newcommand{\inffinmem}[1]{\inffinmemsymb(#1)}
\newcommand{\allinfmemsymb}{\mathsf{nofinext}}
\newcommand{\allinfmem}[1]{\allinfmemsymb(#1)}
\newcommand{\exfinmemsymb}{\mathsf{exfinext}}
\newcommand{\exfinmem}[1]{\exfinmemsymb(#1)}

\newcommand{\memsymb}{{\sf mem}} 
\newcommand{\mem}[2]{\memsymb(#1,#2)}

\newcommand{\inhabproblem}{\mathsf{INHAB}}
\newcommand{\finhabproblem}{\mathsf{FINHAB}}
\newcommand{\inhab}[2]{{\mathcal I}(#1,#2)}

\newcommand{\cntsymb}{\#}
\newcommand{\cnt}[1]{\cntsymb(#1)}

\newcommand{\tnat}{{\mathbb N}}  

\newcommand{\finextsymb}{{\mathcal{E}_\mathrm{fin}}} 
\newcommand{\finext}[1]{\finextsymb(#1)}

\newcommand{\powerfin}[1]{\mathcal{P}_\mathsf{fin}(#1)}
\newcommand{\setoffinmemsymb}{\mathcal{C}}
\newcommand{\setoffinmem}[1]{\setoffinmemsymb(#1)}

\newcommand{\solletter}{\mathcal S}
\newcommand{\solfunction}[1]{{\solletter}(#1)}
\newcommand{\sol}[2]{{\solletter}(\seq{#1}{#2})}

\newcommand{\vect}[1]{\overrightarrow{#1}} 

\newcommand{\impl}{\supset}
\newcommand{\cosign}{\textit{co}}
\newcommand{\fix}{{\nu}}  
\newcommand{\gfpsymb}{\mathsf{gfp}}
\newcommand{\gfp}{\gfpsymb\kern0.1em}  
\newcommand{\oo}{\mathbb{O}}
\newcommand{\lb}{\lambda}

\newcommand{\RIntro}{\textit{RIntro}}

\newcommand{\LVecIntro}{\textit{LVecIntro}}
\newcommand{\Alts}{\textit{Alts}}

\newcommand{\forsomej}{}

\newcommand{\ol}{{\lambda}}
\newcommand{\cool}{{\lambda}^\cosign}
\newcommand{\coolfs}{{\lambda}^\cosign_{\Sigma}}
\newcommand{\coolfsh}{\mathsf{H}{\lambda}^\cosign_{\Sigma}}
\newcommand{\olfsfix}{{\lambda}^\gfpsymb_{\Sigma}}
\newcommand{\olfsfixh}{\mathsf{H}\lambda^\gfpsymb_{\Sigma}}

\newcommand{\seqt}[3]{#1\vdash #2:#3}
\newcommand{\seq}[2]{#1\Rightarrow #2}

\newcommand{\tuple}[1]{\langle #1 \rangle}

\newcommand{\fl}[2]{\langle #2\rangle_{#1}}

\newcommand{\ns}[2]{#1+\cdots+#2} 
\newcommand{\s}[2]{\sum\limits_{#1}^{}{#2}} 

\newcommand{\fs}[2]{\s{#1}{#2}}

\newcommand{\interp}[2]{[\![#1]\!]_{#2}}
\newcommand{\interpwe}[1]{[\![#1]\!]}

\newcommand{\interps}[1]{[\![#1]\!]^s}

\newcommand{\finrepsymb}{{\mathcal F}}
\newcommand{\finrep}[2]{\finrepsymb(#1;#2)}
\newcommand{\finrepempty}[1]{\finrepsymb(#1)}

\newcommand{\colbase}{{\sf mem}}
\newcommand{\colr}[2]{\colbase(#1,#2)}
\newcommand{\colebase}{\colbase}
\newcommand{\colra}[2]{\colebase(#1,#2)}
\newtheorem{lemma}{Lemma}  
\newtheorem{theorem}[lemma]{Theorem}  

\newtheorem{definition}[lemma]{Definition} 
\newtheorem{proposition}[lemma]{Proposition}
\newtheorem{corollary}[lemma]{Corollary}

\newcommand{\infinitesymb}{{\sf inf}}
\newcommand{\infinite}[1]{\infinitesymb(#1)}

\newcommand{\allinfsymb}{{\sf nofin}}
\newcommand{\allinf}[1]{\allinfsymb(#1)}
\newcommand{\allinfa}[1]{\allinfsymb(#1)}

\newcommand{\exfinsymb}{{\sf exfin}}
\newcommand{\exfin}[1]{\exfinsymb(#1)}

\newcommand{\binfsymb}{\mbox{$\not\kern-0.152em{\sf EF}$}}         
\newcommand{\binf}[1]{\binfsymb (#1)}
\newcommand{\binfpsymb}[1]{\binfsymb\kern-0.152em_{#1}} 
\newcommand{\binfp}[2]{\binfpsymb{#1} (#2)}
\newcommand{\binfcanindex}{\star}
\newcommand{\binfcansymb}{\binfpsymb\binfcanindex}
\newcommand{\binfcan}[1]{\binfp\binfcanindex{#1}}

\newcommand{\nbinfsymb}{\mathsf{EF}}    
\newcommand{\nbinf}[1]{\nbinfsymb (#1)}
\newcommand{\nbinfpsymb}[1]{\nbinfsymb\kern-0.152em_{#1}} 
\newcommand{\nbinfp}[2]{\nbinfpsymb{#1} (#2)}
\newcommand{\nbinfcansymb}{\nbinfpsymb\binfcanindex}
\newcommand{\nbinfcan}[1]{\nbinfp\binfcanindex{#1}}

\newcommand{\ndinfsymb}{\not\kern-0.152em\diamond_{\infty}}

\newcommand{\FPV}{\mathit{FPV}}
\newcommand{\dom}{\mathit{dom}}

\newcommand{\power}[1]{\mathcal{P}(#1)}
\newcommand{\dualsymb}{\dagger}
\newcommand{\Mcal}{{\mathcal M}}
\newcommand{\Ncal}{{\mathcal N}}

\title{Inhabitation in Simply-Typed Lambda-Calculus\\ through a Lambda-Calculus for Proof Search}
\author{Jos\'{e} Esp\'{\i}rito Santo, Ralph Matthes\thanks{This work was partially supported by the project \emph{Climt}, ANR-11-BS02-016, of the French Agence Nationale de la Recherche.}, Lu\'{\i}s Pinto}

\begin{document}

\maketitle

\begin{abstract}

A new, comprehensive approach to inhabitation problems in simply-typed lambda-calculus is shown, dealing with both decision and counting problems. This approach works by exploiting a representation of the search space generated by a given inhabitation problem, which is in terms of a lambda-calculus for proof search that the authors developed recently. The representation may be seen as extending the Curry-Howard representation of proofs by lambda-terms, staying within the methods of lambda-calculus and type systems. Our methodology reveals inductive descriptions of the decision problems, driven by the syntax of the proof-search expressions, and the end products are simple, recursive decision procedures and counting functions.

 \end{abstract}


\section{Introduction}\label{sec:intro}

In this paper we study inhabitation problems in the simply-typed $\lb$-calculus, by which we mean both decision problems, like ``does type $A$ have an inhabitant?'', and related questions like counting or listing the inhabitants of a type known to have finitely many of them~\cite{HindleyBasicSimple}. We propose a new approach based on a $\lambda$-calculus for proof search that the authors developed recently~\cite{FICS2013,proofsearch-arXiv}. This is a $\lb$-calculus with fixed-points and formal sums, here named $\olfsfix$, able to represent as a single term the entire space generated by the search for inhabitants for a given type.

Our previous work showed the correctness of this representation. We could add that such representation has a special status: it was derived as an inductive, finitary counterpart to the coinductive characterization of the search process; and the latter is a rather natural (we might say canonical) mathematical definition of the process which, in addition, may be seen as extending the Curry-Howard paradigm of representation, from proofs to runs of search processes (all this will be recalled in Section~\ref{sec:background}). Furthermore, the finitary representation stays within the methods of $\lambda$-calculus and type systems, which dispenses us from importing and adapting methods from other areas, like automata and language theory or games~\cite{TakahashiAH96,BourreauS11,SchubertDB15}, or from creating new representations like in the proof-tree method~\cite{BrodaD05,AlvesB15}.

Despite these formal merits, the applicability of the finitary representation remains to be illustrated. This is the purpose of the present paper. Consider a decision problem~$D$ and let $A$ be a type (of simply-typed $\lambda$-calculus). Our previous work allowed us (i) to express $D(A)$  as $P(S_A)$, where $S_A$ is the coinductive description of the search for inhabitants of $A$, and $P$ is some coinductive predicate, and then (ii) to convert to the equivalent $P'(F_A)$, where $F_A$ is the finitary description of $S_A$ and $P'$ is still a predicate defined by reference to coinductive structures. The form $P'(F_A)$ does not yet profit from the finitary description. This is what we achieve in the present paper: one obtains the equivalent $P''(F_A)$, where $P''$ is inductive, actually directed by the syntax of the finitary description and, for this reason, immediately decidable (however, this might call, at the leaves of the inductive structure, another decidable predicate, whose decidability has been established before with the same method). 

We illustrate in Section~\ref{sec:inhab} the methodology with two decision problems: the problem exemplified above and also ``does type $A$ have finitely many inhabitants?''. Next, in Section~\ref{sec:counting}, we study types $A$ with finitely many inhabitants to show how their number can be calculated from the finitary description $F_A$ as a maybe amazingly simple recursive function.

To sum up: we can base a new methodology to study inhabitation problems on the finitary representation offered by $\olfsfix$ which (i) is aligned with the Curry-Howard isomorphism; (ii) enjoys economy of means, as it finds resources in the area of $\lb$-calculus; (iii) is modular, as it separates the problems of representing the search space (that exploits the subformula property of the object $\lambda$-calculus) from the problem of analyzing it (where types play a minor role beyond being part of the annotations of the deployed $\lambda$-calculus for the analysis); (iv) produces algorithms and functions of high simplicity and even beauty.


\section{Background}\label{sec:background}

This section has four subsections. First we fix our presentation of the simply-typed $\lb$-calculus, next we recall our two representations of proof search, developed before in \cite{FICS2013,proofsearch-arXiv}, and recast here as search for inhabitants of a given type. 
Finally, we start introducing new notions needed in this paper. 

\subsection{Simply-typed $\lb$-calculus}
We lay out a presentation of the simply-typed $\lb$-calculus, a system we often refer to by $\lb$.

Simple types (or simply, types) are given by the grammar:
$$
\begin{array}{lcrcl}
(types)&&A,B,C&:=&p\mid A\impl B
\end{array}
$$
where $p,q,r$ range over \emph{atoms}. We thus do not distinguish types from propositional implicational formulas. We will write $A_1\impl A_2\impl
\cdots\impl A_k\impl p$, with $k\geq0$, in vectorial notation as $\vec{A}\impl p$. For example, if the vector $\vec{A}$ is empty the notation means simply $p$.

Normal (\emph{i.e.}, $\beta$-normal) $\lb$-terms are given by:
$$
\begin{array}{lcrcl}
\textrm{(terms)} &  & t,u & ::= & \,\lambda x^A.t \mid x\,\tuple{t_1,\ldots,t_k}\\
\end{array}
$$
where a countably infinite set of variables, ranged over by letters
$x$, $y$, $w$, $z$, is assumed. Note that in $\lambda$-abstractions we
adopt a \emph{domain-full} presentation (a.\,k.\,a.~Church-style
syntax), annotating the bound variable with a formula. As is common-place with lambda-calculi, we will throughout identify terms up to $\alpha$-equivalence.

As always, we permanently need access to the head variable of a non-abstraction. To this end, we are using an
informal notation, with \emph{vectors} written $\tuple{t_1,\ldots,t_k}$ (meaning $\tuple{}$ if
$k=0$), abbreviated $\tuple{t_i}_i$ if there is no ambiguity on
the range of indices\footnote{If we formalized vectors as a separate syntactic class, with a nil vector and a vector constructor, we would get $\overline{\lb}$-terms \cite{HerbelinCSL94} and would fall, logically, in a sequent calculus format, as in \cite{HerbelinCSL94,proofsearch-arXiv}. But even in \cite{proofsearch-arXiv}, despite the concern with proof search in the sequent calculus, the formalization of vectors was of little importance.}. The term constructor $x\,\tuple{t_1,\ldots,t_k}$ is usually called
\emph{application}. When $n=0$ we may simply write the variable $x$.

We will view contexts $\Gamma$ as finite sets of declarations $x:A$,
where no variable $x$ occurs twice. The context $\Gamma,x:A$ is
obtained from $\Gamma$ by adding the declaration $x:A$, and will only
be written
if $x$ is not
declared in $\Gamma$.
Context union is written as concatenation $\Gamma,\Delta$ for contexts $\Gamma$ and $\Delta$ if $\Gamma\cap\Delta=\emptyset$.
The letters $\Gamma$, $\Delta$, $\Theta$ are used to range over contexts, and the notation $\dom(\Gamma)$ stands for the set of variables declared in $\Gamma$. We will write $\Gamma(x)$ for the type associated with $x$ for $x\in\dom(\Gamma)$, hence viewing $\Gamma$ as a function on $\dom(\Gamma)$.
Context inclusion $\Gamma\subseteq\Delta$ is just set inclusion.

The typing rules are in Fig.~\ref{fig:typing-rules-lambda} and derive sequent $\Gamma\vdash t:A$. $\LVecIntro$ presupposes that the indices for the $t_i$ range over $1,\ldots,k$ and that $\vec B=B_1,\ldots,B_k$, for some $k\geq0$. Such obvious constraints for finite vectors will not be spelt out in the rest of the paper. In the particular case of $k=0$, in which $(x:p)\in\Gamma$ is the only hypothesis of $\LVecIntro$, we type variables (with atoms). Note that the conclusion of the $\LVecIntro$ rule is an atomic sequent---hence a typable term will always be in $\eta$-long form.
\begin{figure}[tb]\caption{Typing rules of $\lb$}\label{fig:typing-rules-lambda}
$$
\begin{array}{c}
\infer[\RIntro]{\seqt\Gamma{\lambda x^A.t}{A\impl B}}{\seqt{\Gamma,x:A}
tB}\quad\quad
\infer[\LVecIntro]
 {\seqt{\Gamma}{x\tuple{t_i}_i}{p}}
 {(x:\vec B\impl p)\in\Gamma\quad\forall i,\,\seqt\Gamma{t_i}{B_i}}
\end{array}
$$
\end{figure}

\subsection{Search for inhabitants, coinductively}

We are concerned with a specific kind of search problems: given $\Gamma$ and $A$, to find $t$ such that $\Gamma\vdash t:A$, that is, to find an \emph{inhabitant} of type $A$ in context $\Gamma$. Under the Curry-Howard correspondence, a pair $\Gamma$, $A$ may be seen as a \emph{logical sequent} $\seq\Gamma A$, and searching for an inhabitant of $A$ in context $\Gamma$ is the same as searching for a proof of that sequent\footnote{To be precise, a proof in natural deduction, which is equivalent to a cut-free, sequent-calculus proof in the system $LJT$ \cite{HerbelinCSL94}.}.

Following \cite{FICS2013,proofsearch-arXiv}, we model this search process through the coinductive $\lb$-calculus, denoted $\cool$. The terms of $\cool$, also called \emph{coterms} or \emph{B\"ohm trees}, are given by
$$ M,N ::=_\cosign \lambda x^A.N\,|\,  x \tuple{N_1,\ldots,N_k}\enspace.$$
This is exactly the previous grammar for $\lb$-terms, but read coinductively, as indicated by the index $\cosign$ (still with finite tuples $\tuple{N_i}_i$). The natural notion of equality between coterms is bisimilarity modulo $\alpha$-equivalence. Following mathematical practice, this is still written as plain equality.

In $\cool$, also the typing rules of Fig.~\ref{fig:typing-rules-lambda} have to be interpreted coinductively---but the formulas/types stay inductive and the contexts finite. Following common practice, we will symbolize the coinductive reading of an inference (rule) by the double horizontal line, but we refrain from displaying Fig.~\ref{fig:typing-rules-lambda} again with double lines---a figure where the two inference rules would be called $\RIntro_\cosign$ and $\LVecIntro_\cosign$.
Such system defines when $\seqt\Gamma NA$ holds for a \emph{finite} context $\Gamma$, a B\"ohm tree $N$ and a type $A$.

Suppose $\seqt\Gamma NA$ holds. Then this sequent has a derivation which is a (possibly infinite) tree of sequents, generated by applying the inference rules bottom-up; and $N$ is a (possibly infinite) coterm, which we call a \emph{solution} of $\sigma$, with $\sigma=(\seq\Gamma A)$. Therefore, such derivations are the structures generated by the search process which does not fail, even if it runs forever, and so they subsume proofs; likewise solutions subsume typable terms (so we may refer to the latter as \emph{finite} solutions\footnote{Solutions subsume finite solutions conservatively. In fact, it is easy to prove that, given a $\lb$-term $t$,  $\seqt\Gamma t A$ in $\ol$ iff $\seqt\Gamma t A$ in $\cool$.})---with solutions still representing derivations, even if infinite, following the Curry-Howard paradigm.

The next step is to extend even further the paradigm, representing also the choice points of the search process. To this end, we extend $\cool$ to $\coolfs$, whose syntax is this:
$$
\begin{array}{lcrcl}
\textrm{(terms)} &  & M,N & ::=_\cosign & \lambda x^A.N\,|\, \ns{E_1}{E_n}\\ 
\textrm{(elimination alternatives)} &  & E & ::=_\cosign & x \tuple{N_1,\ldots,N_k}\\
\end{array}
$$
where both $n,k\geq0$ are arbitrary. $T$ ranges over both terms and elimination alternatives. Note that summands cannot be lambda-abstractions. We will often use $\sum_iE_i$ instead of $\ns{E_1}{E_n}$---in generic situations or if the dependency of $E_i$ on $i$ is clear, as well as the number of elements (if this number is 0, we write the sum as $\oo$).

The most natural notion of equality of terms in $\coolfs$ is again bisimilarity modulo $\alpha$-equivalence, but the notation $\sum_iE_i$ already hints that we consider $+$ to be associative (with $\oo$ as its neutral element). We even want to neglect the precise order of the summands and their (finite) multiplicity. We thus consider the sums of elimination alternatives as if they were sets of alternatives, i.\,e., we further assume that $+$ is symmetric and idempotent. As for $\cool$, we just use mathematical equality for this notion of
bisimilarity on expressions of $\coolfs$, and so the sums of
elimination alternatives can plainly be treated as if they were finite
sets of elimination alternatives (given by finitely many elimination
alternatives of which several might be identified through
bisimilarity).

The expressions of $\coolfs$ are also called \emph{B\"ohm forests}---and a B\"ohm tree $M$ is a member of a B\"ohm forest $N$ when the relation $\colr MN$ defined coinductively in Fig.~\ref{fig:collect} holds.
\begin{figure}[tb]\caption{Membership relations}\label{fig:collect}
$$
\begin{array}{c}
\infer=[]{\colr{\lambda x^A.M}{\lambda x^A.N}}{\colr{M}{N}}\quad\quad
\infer=[]{\colra{x\tuple{M_i}_i}{x\tuple{N_i}_i}}{\forall i,\,\colr{M_i}{N_i}}\quad\quad
\infer=[\forsomej]{\colr{M}{\sum_iE_i}}{\colra{M}{E_j}}\\\\
\end{array}
$$
\end{figure}

In the typing system for $\coolfs$, one derives sequents $\seqt{\Gamma}NA$ and $\seqt{\Gamma}Ep$.
The coinductive typing rules are the ones of $\cool$, together with the rule given in
Fig.~\ref{fig:co-lambda-sum}.
\begin{figure}[tb]\caption{Extra typing rule of $\coolfs$ w.\,r.\,t.~$\cool$}\label{fig:co-lambda-sum}
$$
\begin{array}{c}
\infer=[\Alts]{\seqt{\Gamma}{\sum_iE_i}p}{\forall i,\,\seqt\Gamma{E_i}p}
\end{array}
$$
\end{figure}

A typing derivation of $\coolfs$ is a possibly infinite tree of sequents, generated by the bottom-up application of the inference rules, with ``multiplicative'' branching (logically: ``and'' branching) caused by the list of arguments in elimination alternatives, and ``additive'' branching (logically: ``or'' branching) caused by sums---the latter being able to express the alternatives found in the search process when an atom $p$ can be proved by picking different head variables with their appropriate arguments. So, it is no surprise that, with this infrastructure, we can express, as a single B\"ohm forest, the entire \emph{solution space} generated by the search process when applied to given $\Gamma$ and $A$. That B\"ohm forest can be defined as a function $\solletter$ of $\seq\Gamma A$ defined by corecursion as follows:

\begin{definition}[Solution spaces]
\label{def:sol} 
$$
\sol{\Gamma}{\vec A \impl p} :=  \lambda \vec x:\vec A.\fs{{(y:\vec{B}\supset p)\in\Delta}} {y\fl{j}{\sol{\Delta}{B_j}}}\quad\mbox{with }\Delta:=\Gamma,\vec x:\vec A
$$
\end{definition}
The following properties witness the robustness of the definition \cite{FICS2013,proofsearch-arXiv}.
\begin{proposition}[Properties of solution spaces]\label{prop:properties-of-S} The following properties hold.
\begin{enumerate} 
\item Given $\Gamma$ and $A$, the typing  $\seqt\Gamma{\sol\Gamma A}A$ holds in $\coolfs$.
\item
For $N\in\cool$, $\colr N{\sol\Gamma A}$ iff $\seqt\Gamma N A$ in $\cool$.
\item
For $t\in\ol$, $\colr t{\sol\Gamma A}$ iff $\seqt\Gamma t A$ in $\ol$.\label{prop:properties-of-S.3}
\end{enumerate}
\end{proposition}

\subsection{Search for inhabitants, inductively}

Unfortunately, algorithms cannot in general receive B\"ohm forests as input, so the next step is to find an alternative, equivalent, effective representation that works at least for solution spaces. To this end, an extension $\olfsfix$ of $\lb$ is introduced, whose syntax is given by the following grammar (read inductively):
$$
\begin{array}{lcrcl}
\textrm{(terms)} &  & N & ::= & \lambda x^A.N\mid \gfp\,{X^\sigma}.\ns{E_1}{E_n}\mid X^\sigma\\
\textrm{(elimination alternatives)} &  & E & ::= & x \tuple{N_1,\ldots,N_k}\\
\end{array}
$$
where $X$ is assumed to range over a countably infinite set of
\emph{fixpoint variables} (also letters $Y$, $Z$ will range over them), and where, as for $\coolfs$,
both $n,k\geq0$ are arbitrary. We extend our practice established for $\coolfs$ of writing the sums $\ns{E_1}{E_n}$ in the form $\sum_iE_i$ for $n\geq0$. Also the tuples continue to be communicated as $\fl i{N_i}$. As for $\coolfs$, we will identify expressions modulo associativity, symmetry and idempotence of $+$, thus treating sums of elimination alternatives as if they were the set of those elimination alternatives. Again, we will write $T$ for expressions of $\olfsfix$, i.\,e., for terms and elimination alternatives.

In the term formation rules, $\sigma$ in $X^\sigma$ is required to be \emph{atomic}, i.\,e., of the form $\seq\Gamma p$.
Let $\FPV(T)$ denote the set of free occurrences of typed fixed-point variables in
$T$. Perhaps unexpectedly, in $\gfp\,{X^\sigma}.\sum_iE_i$ the fixed-point construction $\gfp$ binds
\emph{all} free occurrences of $X^{\sigma'}$ in the elimination alternatives $E_i$, not
just $X^\sigma$. But we only want this to happen when $\sigma\leq \sigma'$---which means: the context of $\sigma'$ has more declarations than that of $\sigma$, but not with new types. Formally: $\sigma=(\seq\Gamma p)$, $\sigma'=(\seq{\Gamma'}p)$ and $\Gamma\leq\Gamma'$, with the latter meaning $\Gamma\subseteq\Gamma'$ but $|\Gamma|=|\Gamma'|$, and $|\Delta|$ denoting $\{A\mid \exists x,\, (x:A)\in\Delta\}$ for arbitrary context $\Delta$.

In the sequel, when we refer to \emph{finitary terms} we have in mind the expressions of $\olfsfix$. The fixed-point operator is called
$\gfp$ (``greatest fixed point'') to indicate that its semantics is---see below---defined in terms of the \emph{infinitary} syntax $\coolfs$, but there, fixed points are unique. Hence, the reader may just read this as ``the fixed point''.

We now move to the interpretation of expressions of $\olfsfix$ in terms of the coinductive syntax of $\coolfs$ (using the $\fix$ operation on the meta-level to designate unique fixed points). It is done with the help of \emph{environments} $\xi$, which are partial functions from typed fixed-point variables $X^\sigma$ to
(co)terms of $\coolfs$, with domain $\dom(\xi)$ a finite set of typed fixpoint variables \emph{without duplicates}, which means: $X^{\sigma_1},X^{\sigma_2}\in\dom(\xi)\Rightarrow\sigma_1=\sigma_2$.

Some technicalities are needed before giving the interpretation. We say an environment $\xi$ is admissible for an expression $T$ of $\olfsfix$ if, for every $X^{\sigma'}\in\FPV(T)$, there is an $X^\sigma\in\dom(\xi)$ such that $\sigma\leq\sigma'$. It is easy to see that $T$ admits an environment iff it is \emph{regular} in the following sense: if $X$ occurs free in $T$, there is a sequent $\sigma$ that is the minimum of all $\sigma'$ such that $X^{\sigma'}\in\FPV(T)$. Finally, the interpretation is only given for well-bound expressions, where  $T\in\olfsfix$ is \emph{well-bound} if, for any of its subterms $\gfp\,{X^\sigma}.\sum_iE_i$ and any (free) occurrence of $X^{\sigma'}$ in the $E_i$'s,  $\sigma\leq\sigma'$.
\begin{definition}[Interpretation of finitary terms as B\"ohm forests]
For a well-bound expression $T$ of $\olfsfix$, the interpretation $\interp T\xi$ for an environment $\xi$ that is admissible for 
$T$ is given 
by structural recursion on $T$:
$$
\begin{array}{rcll}
\interp{X^{\sigma'}}\xi& = & [\sigma'/\sigma]\xi(X^{\sigma})\quad\textrm{for the unique $\sigma\leq\sigma'$ with $X^\sigma\in\dom(\xi)$}\\
\interp{\gfp\,{X^{\sigma}}.\s{i}{E_i}}\xi&= & \fix\, N.\s i{\interp {E_i}{\xi\cup[X^{\sigma}\mapsto N]}}\\
\interp{\lambda x^A.N}{\xi}& = & \lambda x^A.\interp N\xi\\
\interp{x \fl i{N_i}}\xi&= & x \tuple{{\interp{N_i}\xi}}_i\\
\end{array}
$$
\end{definition}
If $T$ is closed, i.\,e., $\FPV(T)=\emptyset$, then the empty function is an admissible environment for $T$, and we write $\interpwe T$.

The clause for fixpoint variables in the preceding definition has to cope with $\sigma\leq\sigma'$. This is done by adjusting the value $N=\xi(X^{\sigma})$ looked up in the environment with an operation on B\"ohm forests which will add elimination alternatives to the sums in $N$, in order to match the new declarations in $\sigma'$. If $\sigma=(\seq\Gamma p)$ and $\sigma'=(\seq{\Gamma'}p)$, then $[\sigma'/\sigma]N$ is defined to be $[\Gamma'/\Gamma]N$, with the latter given as follows:

\begin{definition}[Co-contraction]\label{def:co-cont-forests}
Let $\Gamma\leq\Gamma'$. For $T$ an expression of  $\coolfs$, we define $[\Gamma'/\Gamma]T$ by
corecursion as follows:
$$
\begin{array}{lcll}
{[}\Gamma'/\Gamma](\lb x^A.N)&=&\lb x^A.[\Gamma'/\Gamma]N\\
{[}\Gamma'/\Gamma]\s i{E_i}&=&\s i{[\Gamma'/\Gamma]E_i}\\
{[}\Gamma'/\Gamma]\big(z\fl i{N_i}\big)&=&z\fl i{[\Gamma'/\Gamma]N_i}&\textrm{if $z\notin dom(\Gamma)$}\\
{[}\Gamma'/\Gamma]\big(z\fl
i{N_i}\big)&=&\kern-1em\s{(w:A)\in\Delta_z}{w}\fl
i{[\Gamma'/\Gamma]N_i}&\textrm{if $z\in dom(\Gamma)$}
\end{array}
$$
where, in the last clause, $A:=\Gamma(z)$ and $\Delta_z:=\{(z:A)\}\cup(\Gamma'\setminus\Gamma)$.
\end{definition}
Co-contraction captures the extension of the solution space when going
from $\sigma$ to some $\sigma'$ with $\sigma\leq\sigma'$:
\begin{lemma}[Solution spaces and co-contraction]\label{lem:solextension}
Let $\sigma\leq\sigma'$. Then $\solfunction{\sigma'}=[\sigma'/\sigma]\solfunction\sigma$.
\end{lemma}

With the finitary calculus and its semantics in place, we can provide an alternative representation $\finrepempty{\sigma}$ of the search space generated by a sequent $\sigma$.
\begin{definition}[Finitary solution space]
\label{def:finrep} Let $\Xi:=\vect{X:\seq{\Theta}q}$ be a vector
of $m\geq 0$ declarations $(X_i:\seq{\Theta_i}{q_i})$ where no
fixpoint variable name and no sequent occurs twice. The specification of $\finrep{\seq\Gamma{\vec
A\impl p}}{\Xi}$ is as follows:

If, for some $1\leq i\leq m$, $p=q_i$ and $\Theta_i\subseteq\Gamma$ and $|\Theta_i|=|\Gamma|\cup\{A_1,\ldots,A_n\}$, then
$$
\finrep{\seq\Gamma{\vec A\impl p}}{\Xi}=\lb z_1^{A_1}\cdots
z_n^{A_n}.X_i^{\sigma}\enspace,$$
where $i$ is taken to be the biggest such index.
Otherwise, 
$$\finrep{\seq\Gamma{\vec A\impl p}}{\Xi}=
\lb z_1^{A_1}\cdots
z_n^{A_n}.\gfp\,{Y^{\sigma}}.{\s{(y:\vec {B}\impl p)\in\Delta}{y \fl
j{\finrep{\seq{\Delta}{B_j}}{\Xi,Y:\sigma}}}}
$$
where, in both cases, $\Delta:=\Gamma,z_1:A_1,\ldots,z_n:A_n$ and $\sigma:=\seq{\Delta}p$.
\end{definition}
$\finrepempty{\sigma}$ denotes $\finrep{\sigma}{\Xi}$ with empty $\Xi$. It can be proved that: (i) $\finrepempty{\sigma}$ is well-defined (the above recursive definition terminates); (ii) $\finrepempty{\sigma}$ is a closed well-bound term.

The semantics into $\coolfs$ of the finitary representation coincides with $\solletter(\sigma)$ \cite{FICS2013,proofsearch-arXiv}.
\begin{theorem}[Equivalence]\label{thm:FullProp}
For any sequent $\sigma$, $\interpwe{\finrepempty{\sigma}}=\solfunction\sigma$.
\end{theorem}

\subsection{The finite extension}\label{sec:finext}

We now introduce some notions pertaining to the present paper, given its focus on finite inhabitants (\emph{i.e.}, $\lb$-terms).

For $T\in\coolfs$, we call \emph{finite extension} of $T$, which we denote by $\finext T$,  the set of the finite members of $T$, i.\,e., $\finext T=\{t\in\ol\mid\mem tT\}$.
We will be mainly interested in the following predicates on B{\"o}hm forests concerning the finite extension:
\begin{itemize}
\item $\exfinmem T$ is defined to hold iff $\finext T$ is nonempty.
\item $\allinfmem T$ is defined to hold iff $\finext T$ is empty.
\item $\finfinmem T$ is defined to hold iff $\finext T$ is finite.
\item $\inffinmem T$ is defined to hold iff $\finext T$ is infinite.
\end{itemize}
The predicates $\exfinmemsymb$ and $\finfinmemsymb$ will be characterized inductively in
Sect.~\ref{sec:emptiness} and in Sect.~\ref{sec:finiteness}, respectively, together with coinductive characterizations
of $\allinfmemsymb$ and $\inffinmemsymb$
by the generic De Morgan's law relating least and greatest fixed points.


\section{The inhabitation problems}\label{sec:inhab}

We will study two decision problems in simply-typed $\lambda$-calculus: the inhabitation problem and the type finiteness problem. First, we lay down the common approach we will adopt.

Given $\Gamma$ and $A$, we will write $\inhab\Gamma A$ for the set of inhabitants of $A$ relative to context $\Gamma$ in $\ol$, i.\,e., for the set $\{t\in\ol\mid \seqt \Gamma tA\;\textrm{in}\;\ol\}$. Recall that this describes the set of $\eta$-long $\beta$-normal terms of ordinary simply-typed $\lambda$-calculus receiving type $A$ in context $\Gamma$.

\emph{The} inhabitation problem in simply-typed $\lambda$-calculus is the problem ``given $\Gamma$ and $A$, is the set $\inhab{\Gamma}{A}$ nonempty?'', called $\inhabproblem$ in this paper. Its negation is called the ``emptiness problem'' 
(as is well-known, the answer to this question does not depend on whether all $\lambda$-terms are considered or only the $\beta$-normal ones or even the $\eta$-long $\beta$-normal terms).
Decidability of the inhabitation problem in simply-typed $\lambda$-calculus is a well-known result (see, e.\,g., \cite{Statman79}).

\begin{lemma}[Characterization of existence of inhabitants in $\ol$]
	\label{lemma:char-inhab}
	There is a $t\in\ol$ such that \mbox{$\seqt \Gamma t A$} in $\ol$ iff $\exfinmem{\interpwe{\finrepempty{\seq\Gamma A}}}$.

\end{lemma}
\begin{proof}
	\begin{tabular}[t]{cll}
		&$\exists t\in\ol$ s.\,t. $\seqt \Gamma t A$ in $\ol$\\
		iff&$\exists t\in\ol$ s.\,t. $\colr t{\sol\Gamma A}$\qquad(Prop.~\ref{prop:properties-of-S}.\ref{prop:properties-of-S.3})\\
		iff&$\exists t\in\ol$ s.\,t. $\colr t{\interpwe{\finrepempty{\seq\Gamma A}}}$\qquad(Theorem~\ref{thm:FullProp})\\
		iff&$\exfinmem{\interpwe{\finrepempty{\seq\Gamma A}}}$\qquad(by definition of $\exfinmemsymb$)
	\end{tabular}\\[-2ex]
\end{proof}

As seen above, the function $\finrepsymb$ is effectively computable,
and it yields closed well-bound finitary terms. The missing link to
deciding $\inhabproblem$ is thus the decision of the problem
``given a closed well-bound term $T$, does $\exfinmem{\interpwe{T}}$
hold?''. Of course, one cannot deal with closed finitary terms $T$ in
isolation and needs to address fixpoint variables properly. Neither
the interpretation function $\interpwe{\cdot}$ nor the predicate $\exfinmemsymb$ are
effective, but we will define in Section~\ref{sec:emptiness} a syntax-directed
predicate $\nbinfsymb$ (more precisely, it will be a predicate $\nbinfpsymb P$
parameterized over a decidable predicate $P$)
on finitary terms that is equivalent to the composition
$\exfinmemsymb\circ\interpwe{\cdot}$, for at least those closed well-bound terms
that arise as $\finrepempty{\sigma}$ for some sequent $\sigma$ (technically,
the restriction will be to proper terms, as defined in Section~\ref{sec:simplesemantics}).
Syntax-directedness immediately entails that the predicate is decidable.

The appeal of our approach is that, once the finitary representation of the corresponding sequent has been built as $\finrepempty{\sigma}$, the decision of inhabitation is achieved through a simple recursive function over the structure of $\olfsfix$-terms,  corresponding to an inductive predicate adequately characterizing non-emptiness of types.

Using the same methodology, we can also reprove a more difficult and
not so well-known  result of inhabitation for simply-typed $\lambda$-calculus, namely, that the problem ``given $\Gamma$ and $A$, is the set $\inhab{\Gamma}{A}$ finite?'' is decidable (see, e.\,g.~\cite{Hirokawa98}). This problem---henceforth called $\finhabproblem$---depends on studying only $\beta$-normal terms; to recall, the inhabitants of our system $\ol$ are $\eta$-long $\beta$-normal simply-typed $\lambda$-terms, for which the problem is studied in the literature \cite{HindleyBasicSimple} (there, in particular, the algorithm by Ben-Yelles~\cite{benyellesthesis}).

\begin{lemma}[Characterization of type finiteness in $\ol$]
	\label{lemma:char-fin-inhab}
	The set of inhabitants $\inhab{\Gamma}{A}$ is finite iff $\finfinmem{\interpwe{\finrepempty{\seq\Gamma A}}}$.
\end{lemma}
\begin{proof}
Analogous to the proof of Lemma \ref{lemma:char-inhab}, following from Prop.~\ref{prop:properties-of-S}.\ref{prop:properties-of-S.3}, Theorem \ref{thm:FullProp} and definition of $\finfinmemsymb$.
\end{proof}

Analogously to the emptiness problem, our method for establishing decidability of $\finhabproblem$ is to define a recursive predicate on finitary terms that is equivalent to the composition
$\finfinmemsymb\circ\interpwe{\cdot}$, for at least those closed well-bound terms
that arise as $\finrepempty{\sigma}$ for some sequent $\sigma$
(with the same technical condition as for the emptiness problem).
This will be the predicate $\FFsymb$ (again, rather a parameterized predicate $\FFpsymb P$),
studied in Sect.~\ref{sec:finiteness}. Again, the appeal of our approach is that, after building the finitary
representation of the corresponding sequent through the
$\finrepempty{\sigma}$ function, $\finhabproblem$ is decided by a simple
function given recursively over the structure of $\olfsfix$-terms, which, however, additionally
uses the previously established decision algorithm for $\inhabproblem$.

In both cases, the problem is of the form $P\circ\solletter$ on sequents,
and thanks to $\solletter=\interpwe{\cdot}\circ\finrepsymb$ and
associativity, we have to decide
$(P\circ\interpwe{\cdot})\circ\finrepsymb$, where $\finrepsymb$ is
already computable. The solution is by proposing a recursive predicate
$P'$ that can step in for $P\circ\interpwe{\cdot}$, as far the image
of $\finrepsymb$ is concerned (specifically, those terms are
well-bound, have no free fixpoint variables and are proper in the sense of
Definition~\ref{def:proper} below).
Finally, the decision is done by deciding $P'\circ\finrepsymb$.

We will carry out the two instances of this programme, but for this,
it will prove useful to simplify our semantics of finitary terms.


\subsection{A simplified semantics}\label{sec:simplesemantics}

We introduce a simplified interpretation of expressions of $\olfsfix$
in terms of the coinductive syntax of $\coolfs$. We now dispense with
environments and adopt a simpler and even possibly ``wrong''
interpretation, which, however, for $\olfsfix$-terms representing
solution spaces will be seen to be equivalent.
\begin{definition}[Simplified interpretation of finitary terms as B\"ohm forests]
For an expression $T$ of $\olfsfix$, the simplified interpretation $\interps T$ is given
by structural recursion on $T$:
$$
\begin{array}{rcll}
\interps{X^{\sigma}}& = &\solfunction{\sigma}\\
\interps{\gfp\,{X^{\sigma}}.\s{i}{E_i}}&= & \s i{\interps {E_i}}\\
\interps{\lambda x^A.N}& = & \lambda x^A.\interps N\\
\interps{x \fl i{N_i}}&= & x \tuple{{\interps{N_i}}}_i\\
\end{array}
$$
\end{definition}
Note that the base case now profits from the sequent annotation at
fixpoint variables, and the interpretation of the $\gfp$-constructor
dispenses with the use of the $\fix$ operation on the meta-level to
designate unique fixed points on $\coolfs$-expressions. Of course,
this may be ``wrong'' according to our understanding of a greatest
fixed point.

Below, we will be specially interested in the finitary terms which
guarantee that a $\gfp X^\sigma$ construction represents the solution
space of $\sigma$.

\begin{definition}[Proper expressions]\label{def:proper}
  An expression $T\in\olfsfix$ is proper if for any of its subterms
  $T'$ of the form $\gfp\,{X^{\sigma}}.\s{i}{E_i}$, it holds that
  $\interps{T'}=\solfunction{\sigma}$.
\end{definition}

This means that an expression $T$ is considered proper if, despite having
used the simplified definition of semantics for the embedded fixed
points, those subterms have the ``proper'' semantics, and this is only
expressed with respect to our main question of representing solution
spaces, hence where for the fixed-point variables, the reference
semantics of solution spaces is assumed, and this is possible since the fixed-point
variables carry the sequent whose solution space they are intended to
represent.

For proper expressions, the simplified semantics agrees with the semantics
we studied before. Of course, this can only make sense for expressions which
have that previous semantics, in other words for well-bound and regular
expressions.

\begin{lemma}\label{lem:simplequal}
  Let $T$ be well-bound and $\xi$ be an admissible environment for $T$
  such that for all $X^\sigma\in\dom(\xi)$:
  $\xi(X^\sigma)=\solfunction\sigma$. If $T$ is proper, then
  $\interp T\xi=\interps T$.
\end{lemma}
We remark that for any regular $T$, there is exactly one such
environment $\xi$. The case of a closed expression $T$ merits stating
a corollary.
\begin{corollary}\label{cor:simplequal}
  For well-bound, closed and proper $T$, $\interpwe T=\interps T$.
\end{corollary}
\begin{proof}(of Lemma~\ref{lem:simplequal}) By induction on
  expressions $T$. The variable case needs
  Lemma~\ref{lem:solextension}, lambda-abstraction and tuples are fine
  by the induction hypothesis. For the $\gfpsymb$ case, it has to be
  shown that $\interps T$ fulfills the fixed-point equation defining
  $\interp T\xi$, which suffices by uniqueness of the solution. The
  induction hypothesis can be applied to the elimination alternatives
  since the extended environment in which they have to be interpreted
  is of the required form, just by $T$ being proper.
\end{proof}

The corollary is sufficient for our purposes since
$\finrepempty\sigma$ is not only well-bound and closed, but also
proper, as will be seen shortly.

\begin{theorem}[Equivalence for simplified semantics]\label{thm:FullProp-simpl}
  Let $\sigma$ be a sequent and $\Xi$ as in Def.~\ref{def:finrep} so
  that $\finrep\sigma\Xi$ exists (in particular, this holds for empty
  $\Xi$).
  \begin{enumerate}
  \item $\finrep\sigma\Xi$ is proper.\label{thm:FullProp-simpl.1}
  \item $\interps{\finrep\sigma\Xi}=\solfunction\sigma$.\label{thm:FullProp-simpl.2}
  \end{enumerate}
\end{theorem}
\begin{proof}
  Both items together by structural induction on the term
  $\finrep\sigma\Xi$. This all goes by unfolding the definitions and
  use of the induction hypothesis (the main case in the proof of
  \ref{thm:FullProp-simpl.1} needs \ref{thm:FullProp-simpl.2} for
  the subterms, so \ref{thm:FullProp-simpl.1} cannot be proven
  separately before \ref{thm:FullProp-simpl.2}, and the main case
  of \ref{thm:FullProp-simpl.2} immediately follows from the main
  case of \ref{thm:FullProp-simpl.1}, so it is better to prove both
  together, although \ref{thm:FullProp-simpl.2} could be proven
  separately before \ref{thm:FullProp-simpl.1}).
\end{proof}
We remark that the proof is a simplification of the proof for
Theorem~\ref{thm:FullProp} given previously \cite{proofsearch-arXiv}.


\subsection{Deciding type emptiness}\label{sec:emptiness}

We introduce predicate $\allinf T$, for $T$ an expression of $\coolfs$
(B\"{o}hm forest), which holds iff $\allinfmem T$, i.\,e., if the
finite extension of $T$ is empty, but it is defined co-inductively in
Fig.~\ref{fig:allinf}, together (but independently) with the inductive
definition of the predicate $\exfin T$ that is supposed to mean the
negation of $\allinf T$, but which is expressed positively as
existence of a finite member (i.\,e., that the finite extension is
non-empty---that $\exfinmem T$ holds).

	\begin{figure}[tb]\caption{$\exfinsymb$ predicate and $\allinfsymb$ predicate}\label{fig:allinf}
		$$
		\begin{array}{c}
                \infer[]{\exfin{\lambda x^A.N}}{\exfin N}\quad\quad
		\infer[]{\exfin{\sum_iE_i}}{\exfin{E_j}}\quad\quad
		\infer[]{\exfin{x\tuple{N_i}_i}}{\forall i,\,\exfin{N_i}}\\[2ex]
 		\infer=[]{\allinf{\lambda x^A.N}}{\allinf N}\quad\quad
		\infer=[]{\allinf{\sum_iE_i}}{\forall i,\,\allinfa{E_i}}\quad\quad
		\infer=[\forsomej]{\allinfa{x\tuple{N_i}_i}}{\allinf{N_j}}
		\end{array}
		$$
	\end{figure}

\begin{lemma}\label{lem:exfin-allinf}
Given a B\"{o}hm forest $T$, $\exfin T$ iff $\allinf T$ does not hold.
\end{lemma}
\begin{proof}
  See the appendix.
\end{proof}

The following lemma shows that the predicate $\allinfsymb$ corresponds to the intended meaning in terms of the finite extension. Additionally, the lemma shows that the negation of $\allinfsymb$ holds exactly for the B\"{o}hm forests which have finite members.

\begin{lemma}[Coinductive characterization]\label{lem:allinfok}
	Given a B\"{o}hm forest $T$. Then, $\allinf T$ iff 
	$\finext{T}$ is empty, i.\,e., $\allinfsymb=\allinfmemsymb$ as sets of B\"{o}hm forests.
\end{lemma}

\begin{proof} First, let $\infinite M$ be defined coinductively, as belonging to the greatest predicate $\infinitesymb$ satisfying
$$\infinite{\lambda x^A.M}\Leftrightarrow\infinite M\quad\mbox{and}\quad\infinite{x\tuple{M_i}_i}\Leftrightarrow\exists j,\,\infinite{M_j}$$
This is a characterization of infinity: for a B\"ohm tree $M$, $M$ is a $\lb$-term iff $\infinite M$ does not hold. Now, the statement of the lemma is equivalent to: $\allinf T$ iff $\infinite M$ for all $M$ s.\,t. $\colr MT$.
The ``only if'' is equivalently to: if $\allinf T$ and $\colr M T$ then $\infinite M$. 
This is provable by coinduction on $\infinitesymb$, using the obvious
$\colr MM$ for $M$ in $\cool$. The ``if'' implication is suitable for coinduction on $\allinfsymb$, and this works smoothly.
\end{proof}
Thus, we are authorized to work with $\allinfsymb$ and $\exfinsymb$ in
place of their ``extensional variants'' $\allinfmemsymb$ and
$\exfinmemsymb$.

Next we turn to finitary representation of solution spaces and
consider the predicate $\nbinf T$, for $T$ an expression in
$\olfsfix$, which should hold when there is a finite solution. It is
not obvious from the outset if free fixpoint variables should be
considered as contributing to these finite solutions. If one already
knows that $\exfin{\solfunction\sigma}$ holds, then it would be reasonable to put
$X^\sigma$ into the predicate $\nbinfsymb$. However, since our aim is
to prove $\exfinsymb\circ\solletter$ decidable through decidability of $\nbinfsymb$,
we cannot base rules for $\nbinfsymb$ on a decision concerning
$\exfinsymb\circ\solletter$. Still, once we established decidability of
$\exfinsymb\circ\solletter$, we could profit from a definition of $\nbinfsymb$ that
is sharp in the sense of containing variables $X^\sigma$ by definition
if and only if $\exfin{\solfunction\sigma}$. And this we will do in
Section~\ref{sec:finiteness}, building more complex predicates from
$\nbinfsymb$.

We therefore consider a parameterized notion $\nbinfpsymb P$ with $P$ a
predicate on sequents and instantiate it twice, with
\begin{itemize}
\item once $P:=\emptyset$, the empty predicate which is trivially decidable, and,
\item once $\exfinsymb\circ\solletter$ is proven decidable, with $P:=\exfinsymb\circ\solletter$.
\end{itemize}
The general proviso on $P$ is decidability of $P$ and that, for all
sequents $\sigma$, $P(\sigma)$ implies $\exfin{\solfunction\sigma}$,
i.\,e., $P\subseteq\exfinsymb\circ\solletter$. This proviso is
trivially satisfied in both instantiations.\footnote{In a previous
  version of this paper, $P$ was accidentally set to the always true
  predicate, in order to solve a problem of extensionality of a
  predicate that was used to deal with $\finhabproblem$. That was an
  error and led to incorrect proofs. We found this out by ourselves,
  but we also received a counterexample from Micha{\l} Ziobro in
  January 2017 which we gratefully acknowledge.} There are still other
meaningful parameter settings, e.\,g., with $P(\sigma)$ iff
$\sigma=(\seq\Gamma p)$ is instance of an axiom, i.\,e., $p\in|\Gamma|$.

The definition of this (parameterized) predicate $\nbinfpsymb P$ is inductive and
presented in the first line of Fig.~\ref{fig:nbinf}, although it is
clear that it could equivalently be given by a definition by recursion
over the term structure. Therefore, the predicate $\nbinfpsymb P$ is
decidable.

\begin{figure}[tb]\caption{$\nbinfpsymb P$ predicate and $\binfpsymb P$ predicate}\label{fig:nbinf}
		$$
		\begin{array}{c}
                \infer[]{\nbinfp P{X^\sigma}}{P(\sigma)}\quad\quad
		\infer[]{\nbinfp P{\lambda x^A.N}}{\nbinfp P N}\quad\quad
		\infer[\forsomej]{\nbinfp P{\gfp X^\sigma.\sum_iE_i}}{\nbinfp P{E_j}}\quad\quad
		\infer[]{\nbinfp P{x\tuple{N_i}_i}}{\forall i,\,\nbinfp P{N_i}}\\[3ex]
                \infer[]{\binfp P{X^\sigma}}{\neg P(\sigma)}\quad\quad
                \infer[]{\binfp P{\lambda x^A.N}}{\binfp P N}\quad\quad
		\infer[]{\binfp P{\gfp X^\sigma.\sum_iE_i}}{\forall i,\,\binfp P{E_i}}\quad\quad
		\infer[\forsomej]{\binfp P{x\tuple{N_i}_i}}{\binfp P{N_j}}

		\end{array}
		$$
	\end{figure}

The inductive characterization of the negation of the predicate $\nbinfpsymb P$ is easy, as all the rules of $\nbinfpsymb P$ are ``invertible'', and is given in the second line of Fig.~\ref{fig:nbinf}.

	\begin{lemma}\label{lem:binf-nbinf}
		For all $T\in\olfsfix$, $\binfp P T$ iff  $\nbinfp P T$ does not hold.
	\end{lemma}

	\begin{proposition}
		[Finitary characterization]\label{prop:allinf}	
                Let $P$ satisfy $P\subseteq\exfinsymb\circ\solletter$ (this is part of the general proviso on $P$).	
                \begin{enumerate}
                \item If $\nbinfp P T$ then $\exfin{\interps T}$.\label{prop:allinf.1}
                \item Let $T\in\olfsfix$ be well-bound and proper.  If\label{prop:allinf.2}
                  $\binfp P T$ and for all $X^\sigma\in\FPV(T)$, $\exfin{\solfunction\sigma}$ implies $P(\sigma)$,
                  then $\allinf{\interps T}$.
 		\end{enumerate}
	\end{proposition}
	\begin{proof}

          \ref{prop:allinf.1}. is proved by induction on the predicate
          $\nbinfpsymb P$ (or, equivalently, on $T$). The base case for fixed-point variables
          needs the proviso on $P$, and all other cases are immediate
          by the induction hypothesis.

          \ref{prop:allinf.2}. is proved by induction on the predicate
          $\binfpsymb P$ (or, equivalently, on $T$)---the case relative to fixpoints is based on $T$ being proper and needs an
          inner co-induction and also the fact that $\exfinsymb$ is
          invariant under co-contraction. For details, see the appendix.
 	\end{proof}

\begin{theorem}[Decidability of existence of inhabitants in $\ol$]
	\label{thm:decide-fin-inhab}\quad
	\begin{enumerate}
		\item \label{thm:decide-fin-inhab.1}
                  Let $P$ satisfy $P\subseteq\exfinsymb\circ\solletter$ (this is part of the general proviso on $P$).
                  For any $T\in\olfsfix$  well-bound, proper and closed,  $\nbinfp P T$ iff  $\exfin{\interps T}$.
                \item $\exfin{\solfunction\sigma}$ is decidable, by deciding $\nbinfp \emptyset{\finrepempty\sigma}$.\label{thm:decide-fin-inhab.2}
		\item In other words, $\inhabproblem$ is decidable.\label{thm:decide-fin-inhab.3}
	\end{enumerate}
\end{theorem}
\begin{proof} \ref{thm:decide-fin-inhab.1}. Follows from both parts of Prop.~\ref{prop:allinf},
  Lemmas~\ref{lem:exfin-allinf} and \ref{lem:binf-nbinf}, and the
  fact that, trivially, the extra condition in Prop.~\ref{prop:allinf}.\ref{prop:allinf.2} is satisfied for closed
  terms.
	
   \ref{thm:decide-fin-inhab.2}. Apply \ref{thm:decide-fin-inhab.1}.~with both parts of Theorem~\ref{thm:FullProp-simpl}.

   \ref{thm:decide-fin-inhab.3}.
   Analogously to the proof of Lemma~\ref{lemma:char-inhab}, apply
   Prop.~\ref{prop:properties-of-S}.\ref{prop:properties-of-S.3}~(and
   $\exfinmemsymb=\exfinsymb$).
\end{proof}

\begin{definition}
  Let the predicates $\nbinfcansymb$ and $\binfcansymb$ on
  $\olfsfix$ be defined by $\nbinfcansymb:=\nbinfpsymb P$ and
  $\binfcansymb:=\binfpsymb P$ for $P:=\exfinsymb\circ\solletter$,
  which satisfies the proviso by
  Theorem~\ref{thm:decide-fin-inhab}.\ref{thm:decide-fin-inhab.2}. In particular,
  $\nbinfcansymb$ and $\binfcansymb$ are decidable.
\end{definition}

Prop.~\ref{prop:allinf}.\ref{prop:allinf.2} gives that $\binfcan T$ implies $\allinf{\interps T}$ for all well-bound and proper expressions $T$.
However, an inspection of the proof of that lemma even shows that the latter two properties are not needed:

\begin{lemma}[Sharp finitary characterization]\label{prop:allinf-sharp}
  For all $T\in\olfsfix$, $\nbinfcan T$ iff  $\exfin{\interps T}$.
\end{lemma}
\begin{proof}
  See the appendix.
\end{proof}
In particular, $\exfin{\interps T}$ is decidable, by deciding $\nbinfcan T$.


\subsection{Deciding type finiteness}\label{sec:finiteness}

Now a second and more difficult instance of the programme laid out in the beginning of Section \ref{sec:inhab}.

We will now characterize the
predicate $\finfinmemsymb$ by an inductively defined predicate
$\finfinsymb$. Generically, we will obtain a characterization of its
negation $\inffinmemsymb$ by the coinductively defined dual
$\inffinsymb$ of $\finfinsymb$. The inductive definition of $\finfinsymb$ is given in the first line of Fig.~\ref{fig:finfin}.
Notice that, while $\finfinsymb$ is inductively defined and has only finitely many premisses in each clause,
there is absolutely no claim on decidability since the coinductively defined predicate $\allinfsymb$ enters the premisses.

\begin{figure}[tb]\caption{$\finfinsymb$ predicate and $\inffinsymb$ predicate}\label{fig:finfin}
$$
\begin{array}{c}
\infer[]{\finfin{\lambda x^A.N}}{\allinf N}\quad\quad
\infer[]{\finfin{\lambda x^A.N}}{\finfin N}\quad\quad
\infer[]{\finfin{\sum_iE_i}}{\forall i,\,\finfin{E_i}}\quad\quad
\infer[\forsomej]{\finfin{x\tuple{N_i}_i}}{\allinf{N_j}}\quad\quad
\infer[]{\finfin{x\tuple{N_i}_i}}{\forall i,\,\finfin{N_i}}\\[3ex]
\infer=[]{\inffin{\lambda x^A.N}}{\exfin{N}\quad\inffin N}\quad\quad
\infer=[\forsomej]{\inffin{\sum_iE_i}}{\inffin{E_j}}\quad\quad
\infer=[\forsomej]{\inffin{x\tuple{N_i}_i}}{\forall i,\,\exfin{N_i}\quad\inffin{N_j}}
\end{array}
$$
\end{figure}
By inversion (decomposing the summands into tuples) on $\allinfsymb$,
one can show that $\allinfsymb\subseteq\finfinsymb$ (which corresponds semantically to the trivial $\allinfmemsymb\subseteq\finfinmemsymb$). Thus, in particular, no clause
pertaining to $\allinfsymb$ is necessary for the definition of
$\finfin{\sum_iE_i}$.
We now show that $\finfinsymb$ is sound and complete in terms of membership.
\begin{lemma}[Coinductive characterization]\label{lem:finfinok} Given a B\"{o}hm forest $T$. Then, $\finfin T$ iff
$\finext{T}$ is finite, i.\,e., $\finfinsymb=\finfinmemsymb$ as sets of B\"{o}hm forests.
\end{lemma}
\begin{proof}
  The direction from left to right (``soundness'') is immediate by induction on $\finfinsymb$, using Lemma~\ref{lem:allinfok}. From right to left, we do induction on the sum of the term heights of all finite members, which is a finite measure. The first and fourth rule of $\finfinsymb$ are necessary to capture the cases when one passes from $\lambda$-abstractions to their bodies resp.~from tuples to their components---thus when the individual heights decrease---but when there is just no element whose height decreases. The case of sums of elimination alternatives needs a further decomposition into tuples, in order to be able to apply the inductive hypothesis.
\end{proof}
Combined with Lemma \ref{lem:allinfok}, this gives an alternative proof of $\allinfsymb\subseteq\finfinsymb$.

The announced coinductive definition $\inffinsymb$ that is meant to characterize $\inffinmemsymb$ is found in the second line of Fig.~\ref{fig:finfin}.

\begin{lemma}\label{lem:finfin-inffin}
Given a B\"{o}hm forest $T$, $\finfin T$ iff $\inffin T$ does not hold.
\end{lemma}
\begin{proof}
  See the appendix.
\end{proof}

As a corollary, we obtain $\inffinsymb=\inffinmemsymb$ as sets of B{\"o}hm forests.

Now we introduce two predicates on expressions of $\olfsfix$ which will allow to characterize type finiteness, with the following intuitive meanings:
\begin{enumerate}
	\item $\FFp P T$: there are only finitely many finite members of $T$ (the case of no finite members is included in this formulation);
	\item $\NFFp P T$: there are infinitely many finite members of $T$.
\end{enumerate}
Here, the predicate $P$ on sequents controls the case of fixpoint
variables, as before for $\nbinfpsymb P$ and $\binfpsymb P$. The
general proviso on $P$ is that it is decidable and that for all
sequents $\sigma$, $P(\sigma)$ implies $\finfin{\solfunction \sigma}$, i.\,e., $P\subseteq\finfinsymb\circ\solletter$.
For our main result, it will be sufficient to take $P:=\emptyset$.
In view of the decidability result of the previous section,
another possibility of choosing the predicate would be with $P:=\allinfsymb\circ\solletter$,
i.\,e., with the negation of the predicate underlying the definition of $\nbinfcansymb$ and $\binfcansymb$.\footnote{For
this specific setting of $P$, we could easily establish $\NFFpsymb P\subseteq\nbinfcansymb$ or, equivalently,
$\binfcansymb\subseteq\FFpsymb P$, by induction. This would allow 
to remove the condition $\nbinfcan{N_j}$ from the tuple rule for $\NFFpsymb P$.}

The definitions of these predicates are inductive, and they are
presented in Fig.~\ref{fig:FF}. Analogously to the predicates
$\nbinfpsymb P$ and $\binfpsymb P$, it is clear that they could
equivalently be defined recursively over the term structure, thus
ensuring their decidability, thanks to decidability of
$\nbinfcansymb$.

	\begin{lemma}\label{lem:FF-NFF}
		For all $T\in\olfsfix$, $\NFFp P T$ iff $\FFp P T$ does not hold.
	\end{lemma}
	\begin{proof}
		Routine induction on $T$, using Lemma~\ref{lem:binf-nbinf}.
	\end{proof}

\begin{figure}[tb]
	\caption{$\FFpsymb P$ predicate and $\NFFpsymb P$ predicate}
	\label{fig:FF}
	$$
	\begin{array}{c}
        \infer[]{\FFp P{ X^\sigma}}{P(\sigma)}\quad\quad
	\infer[]{\FFp P{\lambda x^A.N}}{\FFp P N}\quad\quad
	\infer[]{\FFp P{\gfp X^\sigma.\sum_iE_i}}{\forall i,\,\FFp P{E_i}}\quad\quad
	\\[2ex]
	\infer[]{\FFp P{x\tuple{N_i}_i}}{\forall i,\,\FFp P{N_i}}\qquad
	\infer[]{\FFp P{x\tuple{N_i}_i}}{\binfcan{N_j}}\quad\qquad\mbox{and}\quad\qquad
	\infer[]{\NFFp P{x\tuple{N_i}_i}}{
		\NFFp P{N_j}\quad\forall i,
		\, \nbinfcan{N_i}}\\[2ex]
         \infer[]{\NFFp P{ X^\sigma}}{\neg P(\sigma)}\quad\quad
	\infer[]{\NFFp P{\lambda x^A.N}}{\NFFp P N}\quad\quad
        \infer[]{\NFFp P{\gfp X^\sigma.\sum_iE_i}}{\NFFp P{E_j}}

	\end{array}
	$$
\end{figure}

\begin{proposition}[Finitary characterization]\label{prop:FF-and-NFF-correct}
Let $P$ satisfy $P\subseteq\finfinsymb\circ\solletter$ (this is part of the general proviso on $P$).
\begin{enumerate}
\item If $\FFp P T$ then $\finfin{\interps T}$.\label{prop:FFimpliesfinfin}
\item Let $T\in\olfsfix$ be well-bound and proper. 
      If $\NFFp P T$ and for all $X^\sigma\in\FPV(T)$, $\finfin{\solfunction\sigma}$ implies $P(\sigma)$,
      then $\inffin{\interps T}$.\label{prop:NFFimpliesinffin}
\end{enumerate}
\end{proposition}
\begin{proof}
  Both statements are proven by induction on $T$ (or, equivalently, by induction on the respective predicate in the premiss).
  While \ref{prop:FFimpliesfinfin}. is straightforward, for \ref{prop:NFFimpliesinffin}.
  the case relative to fixpoints is based on $T$ being proper and needs an
          inner co-induction and also the fact that $\finfinsymb$ is
          invariant under co-contraction. For details (on both parts), see the appendix.
\end{proof}

With these preparations in place, the problem $\finhabproblem$ can be solved in the same way as $\inhabproblem$.

\begin{theorem}[Decidability of type finiteness in $\ol$]\label{thm:decide-finhab-finite}\quad
	\begin{enumerate}
		\item \label{thm:decide-finhab-finite.1}
                  Let $P$ satisfy $P\subseteq\finfinsymb\circ\solletter$ (this is part of the general proviso on $P$).
                  For any $T\in\olfsfix$  well-bound, proper and closed,  $\FFp P T$ iff  $\finfin{\interps T}$.
                \item $\finfin{\solfunction\sigma}$ is decidable, by deciding $\FFp \emptyset{\finrepempty\sigma}$.\label{thm:decide-finhab-finite.2}
		\item In other words, $\finhabproblem$ is decidable.\label{thm:decide-finhab-finite.3}
	\end{enumerate}
\end{theorem}
\begin{proof} \ref{thm:decide-finhab-finite.1}. Follows from both parts of Prop.~\ref{prop:FF-and-NFF-correct},
  Lemmas~\ref{lem:finfin-inffin} and \ref{lem:FF-NFF}, and the
  fact that, trivially, the extra condition in Prop.~\ref{prop:FF-and-NFF-correct}.\ref{prop:NFFimpliesinffin} is satisfied for closed
  terms.
	
   \ref{thm:decide-finhab-finite.2}. Apply \ref{thm:decide-finhab-finite.1}.~with both parts of Theorem~\ref{thm:FullProp-simpl}.

   \ref{thm:decide-finhab-finite.3}.
   Analogously to the proof of Lemma~\ref{lemma:char-fin-inhab}, apply
   Prop.~\ref{prop:properties-of-S}.\ref{prop:properties-of-S.3}~(and
   $\finfinmemsymb=\finfinsymb$).
\end{proof}

\begin{definition}
  Let the predicates $\FFcansymb$ and $\NFFcansymb$ on
  $\olfsfix$ be defined by $\FFcansymb:=\FFpsymb P$ and
  $\NFFcansymb:=\NFFpsymb P$ for $P:=\finfinsymb\circ\solletter$,
  which satisfies the proviso by
  Theorem~\ref{thm:decide-finhab-finite}.\ref{thm:decide-finhab-finite.2}. In particular,
  $\FFcansymb$ and $\NFFcansymb$ are decidable.
\end{definition}

Prop.~\ref{prop:FF-and-NFF-correct}.\ref{prop:NFFimpliesinffin} gives that $\NFFcan T$ implies $\inffin{\interps T}$ for all well-bound and proper expressions $T$.
Again (as for Lemma~\ref{prop:allinf-sharp}), an inspection of the proof of that proposition
even shows that the latter two properties are not needed:

\begin{lemma}[Sharp finitary characterization]\label{prop:inffin-sharp}
  For all $T\in\olfsfix$, $\FFcan T$ iff  $\finfin{\interps T}$.
\end{lemma}
In particular, $\finfin{\interps T}$ is decidable, by deciding $\FFcan T$.


\section{Counting normal inhabitants}\label{sec:counting}

The method of the preceding section is not confined to the mere decision problems. In particular, instead of only deciding $\finhabproblem$, the finitely many inhabitants can be effectively obtained. We will illustrate this with some detail for the somehow more basic question of determining their number. The function for obtaining the set of inhabitants then follows the same pattern.

\smallskip
We have considered B\"ohm forests throughout the paper modulo idempotence of
the summation operation (among other identifications). This does not
hinder us from counting the number of finite members in case it is
finite. The finite members themselves are ``concrete'', and the only
identification that is not expressed in the grammar of $\ol$ is
$\alpha$-equivalence. However, we would prefer counting summand-wise
and thus need to be sure that finite members do not belong to more
than one summand in a sum, and this by taking into account that
occurrences are identified up to bisimulation. Technically, this
desideratum is achieved by considering a subset of B\"ohm forests that
we call \emph{head-variable controlled}. The set $\coolfsh$ of
head-variable controlled B\"ohm forests is obtained by the same grammar
of terms and elimination alternatives as $\coolfs$, but with the
restriction for the formation of $\sum_iE_i$ with
$E_i=x_i\tuple{N^i_j}_j$ that the $x_i$ are pairwise different,
i.\,e., no variable is head of two summands in one sum, and this
recursively throughout the B\"ohm forest. If we consider this
restriction in our view of sums as sets of elimination alternatives, this only
means that a given head variable cannot appear with two distinct
tuples of arguments but still can appear multiply. So, in order to
profit from the extra property of B\"ohm forests in $\coolfsh$, we
regard sums as functions from a finite set of (head) variables $x$
into finite tuples of B\"ohm forest headed by $x$ and use the associated
notion of bisimilarity (modulo $\alpha$-equivalence). This means, when
we speak about head-variable controlled B\"ohm forests, we not only
consider B\"ohm forests satisfying this extra property, but also their
presentation in this form that takes profit from it. This change of view does not
change the notion of bisimilarity. Notice that $\solfunction\sigma$ and $\finrepempty\sigma$ always yield
head-variable controlled terms, in the respective term systems.

We define the counting function $\cntsymb$ for head-variable
controlled B\"ohm forests in $\finfinsymb$ only, by recursion on
$\finfinsymb$.
\begin{definition}[Infinitary counting function $\cntsymb:\coolfsh\cap\finfinsymb\to\tnat$]\label{def:infcnt}
$$
\begin{array}{rcl}
\cnt{\lambda x^A.N}&:=&\left\{ 
\begin{array}{ll}
0& \mbox{if}\;\allinf N \\
\cnt{N}& \mbox{else}\;
\end{array}
\right.\\
\cnt{\sum_iE_i}&:=& \sum_i\; \cnt{E_i}\\
\cnt{x\tuple{N_i}_i}&:=&\left\{ 
\begin{array}{ll}
0& \mbox{if}\;\exists j,\,\allinf{N_j} \\
\prod_i\; \cnt{N_i}& \mbox{else}\;
\end{array}
\right.
\end{array}
$$
\end{definition}

\begin{lemma}\label{lem:allinf-cnt-null}
   Let $T\in\coolfsh$. If $\allinf T$ (in particular, $\finfin
T$) then $\cnt T=0$.
\end{lemma}
\begin{proof}
  Neither induction on $T$ nor on $\allinfsymb$ are available. The proof is by case analysis, where 
  one has to use that elimination alternatives are tuples.
\end{proof}
While this lemma might allow to remove the case
distinction in the $\lambda$-abstraction case, the second branch of
the tuple case would replace the first one only with a very non-strict
reading of the product that would have to be defined and be of value $0$
as soon as one of the factors is $0$.

The following lemma can be considered a refinement of the soundness part of Lemma~\ref{lem:finfinok}.
\begin{lemma}\label{lem:cntok}
  Let $T$ be a head-variable controlled B\"ohm forest such that $\finfin
  T$. Then, $\cnt T$ is a well-defined natural number, and it is the
  cardinality of $\finext T$.
\end{lemma}
\begin{proof}
  Notice that the clause for sums of elimination alternatives is
  subject to the presentation we convened for elements of $\coolfsh$,
  and thus the value is invariant under our identifications. The
  recursive calls to $\cntsymb$ occur only with B\"ohm forests that
  enter $\finfinsymb$ ``earlier''. Being the correct number depends on
  Lemma~\ref{lem:allinfok}.
\end{proof}

Since we have also considered the elements of $\olfsfix$ throughout
the paper modulo idempotence of the summation operation, we will
analogously introduce the set $\olfsfixh$ of head-variable controlled
elements. Again, this is not only a subset but comes with a different
presentation of sums as functions from a finite set of (head)
variables $x$ into finite tuples of finitary terms headed by
$x$. 

\begin{definition}[Finitary counting function $\cntsymb:\olfsfixh\to\tnat$] Define by recursion over the term structure
$$
\begin{array}{rcl}
\cnt{X^\sigma}&:=&0\\
\cnt{\lambda x^A.N}&:=&
\cnt{N}\\
\cnt{\gfp X^\sigma.\sum_iE_i}&:=&
\sum_i\; \cnt{E_i}\\
\cnt{x\tuple{N_i}_i}&:=&
\prod_i\; \cnt{N_i}
\end{array}
$$
\end{definition}

\begin{lemma}\label{lem:binf-cnt-null}
  Let $T\in\olfsfixh\cap\binfcansymb$. Then $\cnt T=0$.
\end{lemma}
\begin{proof}
  Obvious induction, see the appendix.
\end{proof}

\begin{proposition}\label{lem:cntsemantics}
Let $P\subseteq\allinfsymb\circ\solletter$ and $T\in\olfsfixh\cap\FFpsymb P$. Then $\cnt T=\cnt{\interps T}$.
\end{proposition}
\begin{proof}
  The proof is by induction on $T$ (equivalently, by induction on $\FFpsymb P$),
  using Lemma~\ref{lem:binf-cnt-null} for the last rule of $\FFpsymb P$, see the appendix.
\end{proof}

\begin{theorem}[Counting theorem]\label{thm:counting}
Let $P\subseteq\allinfsymb\circ\solletter$ (e.\,g., $P=\emptyset$).
If $\FFp P{\finrepempty\sigma}$ then $\cnt{\finrepempty\sigma}$ is the cardinality of $\finext{\solfunction\sigma}$.
\end{theorem}
\begin{proof}
  $\finrepempty\sigma\in\olfsfixh$. By the
  preceding proposition, using the assumption that
  $\FFp P{\finrepempty\sigma}$, we obtain
  $\cnt{\finrepempty\sigma}=\cnt{\interps{\finrepempty\sigma}}$,
  which is $\cnt{\solfunction\sigma}$ by
  Theorem~\ref{thm:FullProp-simpl}. Thanks to
  Proposition~\ref{prop:FF-and-NFF-correct}.\ref{prop:FFimpliesfinfin},
  $\finfin{\solfunction\sigma}$, hence, by Lemma~\ref{lem:cntok},
  $\cnt{\solfunction\sigma}$ is the cardinality of
  $\finext{\solfunction\sigma}$.
\end{proof}
Notice that when $\FFp P{\finrepempty\sigma}$ does not hold, then
$\cnt{\finrepempty\sigma}$ is meaningless, but
$\NFFp P{\finrepempty\sigma}$ holds, and thus,
$\inffin{\solfunction\sigma}$, which ensures an infinite number of
finite solutions of $\sigma$.

Without any extra effort, we can give an effective definition of the associated set of finite inhabitants through a function $\setoffinmemsymb:\olfsfixh\to\powerfin{\ol}$ by
$$
\begin{array}{rcl}
\setoffinmem{X^\sigma}&:=&\emptyset\\
\setoffinmem{\lambda x^A.N}&:=&
\{\lambda x^A.t\mid t\in\setoffinmem{N}\}\\
\setoffinmem{\gfp X^\sigma.\sum_iE_i}&:=&
\cup_i\;\setoffinmem{E_i}\\
\setoffinmem{x\tuple{N_i}_i}&:=&\{x\tuple{t_i}_i\mid \forall i,\,t_i\in\setoffinmem{N_i}\}
\end{array}
$$
Then, for $T\in \olfsfixh$, $\cnt T$ is the cardinality of $\setoffinmem T$ (notice that the set union in the $\gfp$ case is always a disjoint union), and if $\FFp\emptyset{\finrepempty{\seq\Gamma A}}$ then $\inhab\Gamma A=\setoffinmem{\finrepempty{\seq\Gamma A}}$. If not, $\inhab\Gamma A$ is infinite.


\section{Final Remarks}\label{sec:final}

This paper illustrates a methodology to address decidability problems in the simply-typed $\lambda$-calculus which starts by computing a $\lambda$-term (through function $\finrepsymb$) representing the full set of inhabitants of a given type (using an extension of $\lambda$-calculus designed previously by the authors), and then uses that $\lambda$-term to decide the problem at hand.

To carry out this program, we had to introduce our simplified
semantics that is loose in the sense that it does not guarantee that
the interpretation of formal fixed-point constructs indeed denotes a
fixed point. This loose semantics can be analyzed very smoothly, and we
also identified the notion of a proper expression where the simplified
semantics agrees on formal fixed-point constructs with the intended
semantics in terms of solution spaces. Our finitary representation
function generates proper expressions, and so we can apply the
simplified semantics to solve the original problems.

The predicates with which we analyze the finitary expressions
representing sets of inhabitants are parameterized by a predicate on
sequents for the case of fixpoint variables. The interesting point
about our use of this parameter is that, in order to establish
decidability, we choose it very simply (as the empty set), but once we
obtained decidability, we can in turn use that predicate as parameter
when building further definitions. In the end, we only need two
instances, but we consider it important---not only in the interest of
succinctness---to have identified this abstraction.

\medskip
We do not claim that our method can confirm sharp complexity results, namely ${\sf PSPACE}$-completeness of $\inhabproblem$~\cite{Statman79} and $\finhabproblem$~\cite{Hirokawa98}.
We are rather interested in having a simple representation of the \emph{full} sets of inhabitants, which may have multiple uses, as illustrated by our counting functions. By ``full'' we mean in particular that we capture all $\eta$-long and $\beta$-normal terms. The restriction to $\eta$-long terms is very convenient for a concise description and does not do much harm to the usability of the results. The concept of \emph{co-contraction} (Def.~\ref{def:co-cont-forests}) is crucial for completeness of our method in this respect, and as shows our paper, it is not intrusive in practice, i.\,e., for the analysis carried out in this paper, its presence is hardly noticed in the proofs.

Note that other approaches dealing with a full set of inhabitants also face questions. 
For example, in~\cite{TakahashiAH96}, although (finite) context-free grammars suffice to capture inhabitants obeying the \emph{total discharge convention}, an infinite grammar is used to capture the full set of $\beta$-normal forms. In~\cite{BrodaD05} (Sect. 4.3) a method is presented to produce a context-free grammar to generate the long normal forms of a type, but the produced grammars seem again to be unable to stay within the abovementioned optimal complexity.
Also, in~\cite{SchubertDB15}, as a goal is to achieve machines capable of enumerating all normal inhabitants, and for this, storing a fixed finite number of bound variable names is not sufficient, automata with a non-standard form of register are used. 

We believe our compositional methodology of first building a $\lambda$-term (more precisely, a closed well-bound term in $\olfsfix$) representing the full set of inhabitants of interest, and then traversing that $\lambda$-term to decide whether a given property of that set holds, can be transferred to other contexts. For example, it would be interesting to know if in the presence of a connective like disjunction, our methodology produces a (simple) decision function for the $\inhabproblem$ problem.



\newpage
\appendix
\section{Proofs}
\noindent
	{\bf Lemma
	\ref{lem:exfin-allinf}.}
\emph{Given a B\"{o}hm forest $T$, $\exfin T$ iff $\allinf T$ does not hold.}
\begin{proof}
  This is plainly an instance of the generic result in the style of
  De Morgan's laws that presents
  inductive predicates as complements of coinductive predicates, by a
  dualization operation on the underlying clauses. The principle is
  recalled with details now.

  Assume a set $U$ (the
  ``universe'') and a function $F:\power U\to\power U$ that is monotone,
  i.\,e., for $\Mcal\subseteq \Ncal\subseteq U$, one has $F(\Mcal)\subseteq
  F(\Ncal)$. Then, by Tarski's fixed-point theorem, there exist the least
  fixed-point $\mu F$ and the greatest fixed-point $\nu F$ of $F$, with
  respect to set inclusion. Moreover, $\mu F$ is the intersection of all
  pre-fixed points $\Mcal\subseteq U$ of $F$, i.\,e., with $F(\Mcal)\subseteq
  \Mcal$, and $\nu F$ is the union of all post-fixed points $\Mcal\subseteq U$
  of $F$, i.\,e., with $\Mcal\subseteq F(\Mcal)$. This lattice-theoretic duality
  allows to relate both concepts through complements, with $\Mcal^\complement:=U\setminus \Mcal$.  Given $F$ as
  before, define a monotone function $F^\dualsymb:\power U\to\power U$
  by setting $F^\dualsymb(\Mcal):=(F(\Mcal^\complement))^\complement$. Then,
  $$\mu F=(\nu (F^\dualsymb))^\complement\enspace.$$ This formula (written in
  logical terms with negation in place of set complement) is often used
  to \emph{define} $\mu F$, e.\,g., in $\mu$-calculus. For a proof, it
  suffices to consider the inclusion from left to right (the other
  direction is obtained by duality, using $(F^\dualsymb)^\dualsymb=F$).
  Since the left-hand side
  is included in every pre-fixed point of $F$, it suffices to
  show that the right-hand side is such a pre-fixed point, i.\,e.,
  $F((\nu (F^\dualsymb))^\complement)\subseteq(\nu (F^\dualsymb))^\complement$. We
  show the contrapositive $\nu (F^\dualsymb)\subseteq F^\dualsymb(\nu (F^\dualsymb))$ (using
  $F^\dualsymb$ as abbreviation): but $\nu(F^\dualsymb)$ is a post-fixed point itself (it is even a
  fixed point).
\end{proof}
\begin{lemma}\label{lem:invcocontraction}
Let $P\in\{\exfinsymb,\allinfsymb\}$ and $\sigma\leq\sigma'$. Then for all B\"ohm forests $T$, we have $P(T)$ iff $P([\sigma'/\sigma]T)$.
\end{lemma}

\noindent
{\bf Lemma \ref{lem:binf-nbinf}.}
\emph{For all $T\in\olfsfix$, $\binf T$ iff  $\nbinf T$ does not hold.}

	\begin{proof}
		Routine induction on $T$. In terms of the equivalent recursive
		definitions of the predicates, this would have been just an application
		of De Morgan's laws.
	\end{proof}

\noindent
{\bf Proposition \ref{prop:allinf} (Finitary characterization).}
\emph{Let $P$ satisfy $P\subseteq\exfinsymb\circ\solletter$ (this is part of the general proviso on $P$).	
                \begin{enumerate}
                \item If $\nbinfp P T$ then $\exfin{\interps T}$.
                \item Let $T\in\olfsfix$ be well-bound and proper.  If
                  $\binfp P T$ and for all $X^\sigma\in\FPV(T)$, $\exfin{\solfunction\sigma}$ implies $P(\sigma)$,
                  then $\allinf{\interps T}$.
 		\end{enumerate}
}
       \begin{proof}
                \ref{prop:allinf.1}. is sketched in the main part of the paper.

		\ref{prop:allinf.2}. is proved by induction on the predicate $\binfpsymb P$
                 (which can also be seen as a proof by induction on $T$).

		Case $T=X^\sigma$. Then $\neg P(\sigma)$, hence, since $X^\sigma\in\FPV(T)$, by contraposition and
                Lemma~\ref{lem:exfin-allinf}, we get  $\allinf{\solfunction\sigma}$.
		
		Case $T=\gfp X^\sigma.\sum_iE_i$. Let $N:=\interps T=\sum_i \interps{E_i}$. As $T$ is proper, $N=\solfunction{\sigma}$.
                We hence have to show $\allinf{\solfunction{\sigma}}$, which we do by an embedded coinduction
                for the coinductively defined predicate $\allinfsymb$.
                We have $\binfp P{E_i}$ for all $i$ and want to use the induction hypothesis,
                which would give us $\allinf{\interps{E_i}}$ and thus $\allinf{\sum_i \interps{E_i}}$, which was our goal.
                Fix an $i$. Of course, $E_i$ is also well-bound and proper. We have to consider all $Y^{\sigma'}\in\FPV(E_i)$. Either
                $Y^{\sigma'}\in\FPV(T)$, and we are fine by hypothesis, or $Y=X$ and, since $T$ is well-bound, $\sigma\leq\sigma'$.
                We just show that $\exfin{\solfunction{\sigma'}}$ does not hold: from our coinductive hypothesis $\allinf{\solfunction{\sigma}}$,
                we get through Lemma~\ref{lem:solextension} and Lemma~\ref{lem:invcocontraction} even $\allinf{\solfunction{\sigma'}}$,
                and this is the negation of $\exfin{\solfunction{\sigma'}}$. This is a proper application of the coinductive hypothesis
                since it enters a lemma on $\allinfsymb$ that does not change needed observation depths and then goes into an elimination alternative,
                where the occurrences of free fixpoint variables are at least ``guarded'' by an ordinary variable of a tuple.

                The other cases are simple applications of the induction hypothesis.

	\end{proof}

\noindent
{\bf Lemma \ref{prop:allinf-sharp} (Sharp finitary characterization).}
\emph{For all $T\in\olfsfix$, $\nbinfcan T$ iff  $\exfin{\interps T}$.}
\begin{proof}
  In view of the previous proposition, we only need to consider the
  direction from right to left, and we prove its contraposition
  $\binfcan T$ implies $\allinf{\interps T}$ by induction on the predicate $\binfcansymb$.

  Case $T=X^\sigma$. Then $\neg\exfin{\solfunction\sigma}$ by hypothesis of this case, and this is $\allinf{\interps{X^\sigma}}$.

  Case $T=\gfp X^\sigma.\sum_iE_i$. Then $\interps T=\sum_i \interps{E_i}$.
                We have $\binf{E_i}$ for all $i$ and we use the induction hypothesis,
                which gives us $\allinf{\interps{E_i}}$ for all $i$ and thus $\allinf{\sum_i \interps{E_i}}$, which was our goal.
                Notice that this reasoning does not need further properties of $T$.

  The other cases are likewise simple applications of the induction hypothesis.
\end{proof}

\noindent
\textbf{Lemma \ref{lem:finfin-inffin}.}
\emph{Given a B\"{o}hm forest $T$, $\finfin T$ iff $\inffin T$ does not hold.}
\begin{proof}
  $\inffinsymb$ is defined from $\finfinsymb$ by the De Morgan's law (as
  recalled in the proof of Lemma~\ref{lem:exfin-allinf}). In the first
  clause for $\inffinsymb$, the proviso $\exfin N$ is necessary for soundness, and as
  well the proviso $\exfin{N_j}$ (with $i=j$) in the last clause. Only
  through these guards we can ensure that
  $\inffinsymb\subseteq\exfinsymb$, which is a minimum requirement
  given what they say in terms of finite membership. Otherwise, the
  first clause would allow to derive $\inffin N$ for the infinite
  $\lambda$-abstraction, satisfying the equation $N=\lambda x^A.N$ for
  any choice of $A$ and without any relevance of the variable
  $x$. Similarly for the third clause with $\nu N.x\tuple N$.
\end{proof}

\begin{lemma}\label{lem:invcocontraction1}
	Let $P\in\{\finfinsymb,\inffinsymb\}$ and $\Gamma\leq\Gamma'$. Then for all B\"ohm forests $T$, we have $P(T)$ iff $P([\Gamma'/\Gamma]T)$.
\end{lemma}


\noindent
\textbf{Proposition \ref{prop:FF-and-NFF-correct}. (Finitary characterization)}
\emph{Let $P$ satisfy $P\subseteq\finfinsymb\circ\solletter$ (this is part of the general proviso on $P$).
\begin{enumerate}
\item If $\FFp P T$ then $\finfin{\interps T}$.
\item Let $T\in\olfsfix$ be well-bound and proper. 
      If $\NFFp P T$ and for all $X^\sigma\in\FPV(T)$, $\finfin{\solfunction\sigma}$ implies $P(\sigma)$,
      then $\inffin{\interps T}$.
\end{enumerate}}

\begin{proof} \ref{prop:FFimpliesfinfin}. By induction on $\FFpsymb P$ (or equivalently by structural induction on $T$).
              We only show the tuple cases with $T={x\tuple{N_i}_i}$. The other cases are equally simple.

Case for some $j$, $\binfcan{N_j}$. By Lemma~\ref{prop:allinf-sharp}, $\allinf{\interps{N_j}}$, hence $\finfin{x\tuple{\interps{N_i}}_i}$,
which is $\finfin{\interps T}$.

Case for all $i$, $\FFp P{N_i}$. By induction hypothesis, $\finfin{\interps{N_i}}$ for all $i$, hence  $\finfin{\interps{x\tuple{N_i}_i}}$.
	
\ref{prop:NFFimpliesinffin}. By induction on $\FFpsymb P$ (or equivalently by structural induction on $T$).

Case $T=X^\sigma$. Then $\neg P(\sigma)$, hence, since $X^\sigma\in\FPV(T)$, by contraposition and
                Lemma~\ref{lem:finfin-inffin}, we get  $\inffin{\solfunction\sigma}$.

Case $T={x\tuple{N_i}_i}$. For some $j$,
	$\NFFp P{N_j}$ and, for all $i$,
	$\nbinfcan{N_i}$. The induction hypothesis is applicable for $N_j$ since $\FPV(N_j)\subseteq\FPV(T)$.
        Therefore, we have $\inffin{\interps{N_j}}$. By Lemma~\ref{prop:allinf-sharp}, $\exfin{\interps{N_i}}$,
        for all $i$, hence, we are done by definition of $\inffinsymb$.

Case $T=\gfp X^\sigma.\sum_iE_i$. For some $j$, $\NFFp P{E_j}$. Let $N:=\interps T=\sum_i \interps{E_i}$.
                As $T$ is proper, $N=\solfunction{\sigma}$.
                We hence have to show $\inffin{\solfunction{\sigma}}$, which we do by an embedded coinduction
                for the coinductively defined predicate $\inffinsymb$.
                We want to use the induction hypothesis for $E_j$,
                which would give us $\inffin{\interps{E_j}}$ and thus $\inffin{\sum_i \interps{E_i}}$, which was our goal.
                Of course, $E_j$ is also well-bound and proper. We have to consider all $Y^{\sigma'}\in\FPV(E_j)$. Either
                $Y^{\sigma'}\in\FPV(T)$, and we are fine by hypothesis, or $Y=X$ and, since $T$ is well-bound, $\sigma\leq\sigma'$.
                We just show that $\finfin{\solfunction{\sigma'}}$ does not hold: from our coinductive hypothesis $\inffin{\solfunction{\sigma}}$,
                we get through Lemma~\ref{lem:solextension} and Lemma~\ref{lem:invcocontraction1} even $\inffin{\solfunction{\sigma'}}$,
                and this is the negation of $\finfin{\solfunction{\sigma'}}$. This is a proper application of the coinductive hypothesis
                since it enters a lemma on $\inffinsymb$ that does not change needed observation depths and then goes into an elimination alternative,
                where the occurrences of free fixpoint variables are at least ``guarded'' by an ordinary variable of a tuple.

The case of $\lambda$-abstractions is a simple application of the induction hypothesis.
\end{proof}

                We remark that the proposition and its proof are
                rather analogous to Prop.~\ref{prop:allinf} than dual to it, although
                the logical structure of the predicates is rather
                dual: to enter a fixed point into $\FFpsymb P$, all of
                the elimination alternatives have to be there already,
                while for $\nbinfpsymb P$, only one of the elimination
                alternatives is required. However, this duality is
                broken for the tuples: while for $\nbinfpsymb P$, all
                arguments are required to be in the same predicate,
                $\FFpsymb P$ has a rule that asks only about one
                argument, but for a different predicate, and there is
                even a second possibility. Anyway, the proof structure
                needs to be analogous since $\exfinsymb$ and
                $\finfinsymb$ are both inductively defined and
                therefore do not admit reasoning by coinduction.

\medskip
\noindent
\textbf{Lemma \ref{lem:binf-cnt-null}.}
\emph{Let $T\in\olfsfixh\cap\binfcansymb$. Then $\cnt T=0$.}
\begin{proof}
  Induction over $\binfcansymb$ (or, equivalently, over $T$).

  Case $T=X^\sigma$. Trivial.

  Case $T=\lambda x^A.N$. Trivial by induction hypothesis.
  
  Case $T={x\tuple{N_i}_i}$. By induction hypothesis, one of the factors is $0$.

  Case $T=\gfp X^\sigma.\sum_iE_i$. By induction hypothesis, all summands are $0$.
\end{proof}

\medskip
\noindent
\textbf{Proposition \ref{lem:cntsemantics}.}
\emph{Let $P\subseteq\allinfsymb\circ\solletter$ and $T\in\olfsfixh\cap\FFpsymb P$. Then $\cnt T=\cnt{\interps T}$.}
\begin{proof}
  We will write $L$ and $R$ for left-hand side and right-hand side of the equation to prove. The proof is by induction on $T$
  (or, equivalently, by induction on $\FFpsymb P$).

  Case $T=X^\sigma$. Then $\allinf{\solfunction\sigma}$, hence $\cnt{\solfunction\sigma}=0$ by Lemma~\ref{lem:allinf-cnt-null}.
  Hence, $R=0=L$.

  Case $T=\lambda x^A.N$. Then $\FFp PN$. $L=\cnt N$. $R=\cnt{\lambda x^A.\interps N}$. According to the definition of $R$, we have to 
  distinguish if $\allinf{\interps N}$ or not. In the first case, by
  Lemma~\ref{lem:allinf-cnt-null}, we have $\cnt{\interps N}=0$. Thus, in both case, this gives $R=\cnt{\interps N}$,
  while $L=\cnt N$. Done by induction hypothesis.
  
  Case $T={x\tuple{N_i}_i}$. Subcase $\binfcan{N_j}$ for some $j$. By Lemma~\ref{prop:allinf-sharp}, $\allinf{\interps{N_j}}$.
  Hence, $R=0$. By Lemma~\ref{lem:binf-cnt-null}, $\cnt{N_j}=0$, hence also $L=0$ (since one factor is $0$).

  Subcase $\FFp P{N_i}$ for all $i$. We may assume that we are not in the first subcase that has already been treated, hence $\nbinfcan{N_i}$ for all $i$.
  By Lemma~\ref{prop:allinf-sharp}, $\neg\allinf{\interps{N_i}}$ for all $i$. Therefore, $R=\prod_i\; \cnt{\interps{N_i}}$, while $L=\prod_i\; \cnt{N_i}$.
  Done by induction hypothesis for all $i$. 

  Case $T=\gfp X^\sigma.\sum_iE_i$. Then $\FFp P{E_i}$ for all $i$. Just apply the induction hypothesis to all the summands and sum up.
  (Notice how this case becomes the simplest one in our setting with simplified semantics.)
\end{proof}


\end{document}